\documentclass[11pt,letterpaper]{article} 
\usepackage[utf8]{inputenc}
\usepackage[T1]{fontenc}
\usepackage[english]{babel}
\usepackage[svgnames]{xcolor}
\usepackage{amsmath}
\usepackage{amsthm}
\usepackage{amssymb}
\usepackage{mathtools}
\usepackage{comment}
\usepackage{url}
\usepackage[in]{fullpage}
\usepackage{thmtools}
\usepackage{thm-restate}
\usepackage{graphicx}
\usepackage{natbib}
\usepackage{complexity}
\usepackage{multirow}
\usepackage{nicefrac}
\usepackage{etoolbox}
\usepackage{algorithm}
\usepackage[noend]{algpseudocode}
\usepackage[textsize=scriptsize]{todonotes}
\usepackage{hyperref}

\newtoggle{singlecolumn}
\toggletrue{singlecolumn}

\declaretheorem[]{theorem}
\declaretheorem[sibling=theorem]{definition}
\declaretheorem[sibling=theorem]{lemma}

\declaretheorem[sibling=theorem]{fact}

\newcommand{\loss}{\text{loss}}

\renewcommand{\R}{\mathbb{R}}

\newcommand{\tr}{\text{tr}}
\newcommand{\transp}{\mathsf{T}}
\newcommand{\Pclass}{\text{P}}
\newcommand{\NPclass}{\text{NP}}

\newcommand{\val}{\text{value}}
\newcommand{\Hamming}{\text{Hamming}}

\DeclareMathOperator*{\argmin}{arg\,min}

\newcommand{\eps}{\varepsilon}
 
\begin{document}
\title{On the Fine-Grained Complexity\\of Empirical Risk Minimization:\\ Kernel Methods and Neural Networks}
\date{}
\author{
	Arturs Backurs\\ MIT\\\texttt{backurs@mit.edu}
	\and Piotr Indyk\\ MIT\\\texttt{indyk@mit.edu}
	\and Ludwig Schmidt\\ MIT\\\texttt{ludwigs@mit.edu}
	}
\clearpage\maketitle
 
\begin{abstract}
Empirical risk minimization (ERM) is ubiquitous in machine learning and underlies most supervised learning methods.
While there has been a large body of work on algorithms for various ERM problems, the exact computational complexity of ERM is still not understood.
We address this issue for multiple popular ERM problems including kernel SVMs, kernel ridge regression, and training the final layer of a neural network.
In particular, we give conditional hardness results for these problems based on complexity-theoretic assumptions such as the Strong Exponential Time Hypothesis.
Under these assumptions, we show that there are no algorithms that solve the aforementioned ERM problems to high accuracy in sub-quadratic time.
We also give similar hardness results for computing the gradient of the empirical loss, which is the main computational burden in many non-convex learning tasks.
 \end{abstract}

\section{Introduction}
Empirical risk minimization (ERM) has been highly influential in modern machine learning~\cite{vapnik1998statistical}.
ERM underpins many core results in statistical learning theory and is one of the main computational problems in the field.
Several important methods such as support vector machines (SVM), boosting, and neural networks follow the ERM paradigm~\cite{SS14}.
As a consequence, the algorithmic aspects of ERM have received a vast amount of attention over the past decades.
This naturally motivates the following basic question:

\begin{center}
\emph{What are the computational limits for ERM algorithms?}
\end{center}

In this work, we address this question both in convex and non-convex settings.
Convex ERM problems have been highly successful in a wide range of applications, giving rise to popular methods such as SVMs and logistic regression.
Using tools from convex optimization, the resulting problems can be solved in polynomial time.
However, the exact time complexity of many important ERM problems such as kernel SVMs is not yet well understood.
As the size of data sets in machine learning continues to grow, this question is becoming increasingly important.
For ERM problems with millions of high-dimensional examples, even quadratic time algorithms can become painfully slow (or expensive) to run.

Non-convex ERM problems have also attracted extensive research interest, e.g., in the context of deep neural networks.
First order methods that follow the gradient of the empirical loss are not guaranteed to find the global minimizer in this setting.
Nevertheless, variants of gradient descent are by far the most common method for training large neural networks.
Here, the computational bottleneck is to compute a number of gradients, not necessarily to minimize the empirical loss globally.
Although we can compute gradients in polynomial time, the large number of parameters and examples in modern deep learning still makes this a considerable computational challenge.

Unfortunately, there are only few existing results concerning the exact time complexity of ERM or gradient computations.
Since the problems have polynomial time algorithms, the classical machinery from complexity theory (such as NP hardness) is too coarse to apply.
Oracle lower bounds from optimization offer useful guidance for convex ERM problems, but the results only hold for limited classes of algorithms.
Moreover, they do not account for the \emph{cost} of executing the oracle calls,  as they simply lower bound their number. 
Overall, we do not know if common ERM problems allow for algorithms that compute a high-accuracy solution in sub-quadratic or even nearly-linear time for all instances.\footnote{More efficient algorithms exist if the running time is allowed to be polynomial in the accuracy parameter, e.g., \cite{shalev2007pegasos} give such an algorithm for the kernel SVM problem that we consider as well.
See also the discussion at the end of this section.}
Furthermore, we do not know if there are more efficient techniques for computing (mini-)batch gradients than simply treating each example in the batch independently.\footnote{Consider a network with one hidden layer containing $n$ units and a training set with $m$ examples, for simplicity in small dimension $d = O(\log n)$.
No known results preclude an algorithm that computes a full gradient in time $O( (n + m) \log n)$.
This would be significantly faster than the standard $O( n \cdot m \cdot \log n)$ approach of computing the full gradient example by example.}

We address both questions for multiple well-studied ERM problems.
First, we give conditional hardness results for minimizing empirical risk in several settings, including kernel SVMs, kernel ridge regression (KRR), and training the top layer of a neural network.
Our results give evidence that no algorithms can solve these problems to high accuracy in strongly sub-quadratic time.
Moreover, we provide similar conditional hardness results for kernel PCA.
All of these methods are popular learning algorithms due to the expressiveness of the kernel or network embedding.
Our results show that this expressiveness also leads to an expensive computational problem.

Second, we address the complexity of computing a gradient for the empirical risk of a neural network.
In particular, we give evidence that computing the gradient of the top layer in a neural network takes time that is ``rectangular'', i.e., the time complexity cannot be significantly better than $O(n \cdot m)$, 
where $m$ is the number of examples and $n$ is the number of units in the network.
Hence, there are no algorithms that compute batch gradients faster than handling each example individually, unless common complexity-theoretic assumptions fail.

Our hardness results are based on recent advances in fine-grained complexity and build on conjectures such as the Strong Exponential Time Hypothesis (SETH)~\cite{impagliazzo2001problems,impagliazzo2001complexity,vassilevska2015hardness}.
SETH concerns the classic satisfiability problem for formulas in Conjunctive Normal Form (CNF).
Informally, the conjecture states that there is no algorithm for checking satisfiability of a formula with $n$ variables and $m$ clauses in time less than $O(c^n \cdot \mbox{poly}(m))$ for some $c<2$.\footnote{Note that SETH can be viewed as a significant strengthening of the $\Pclass \neq \NPclass$ conjecture, which only postulates that there is no \emph{polynomial time} algorithm for CNF satisfiability. The best known algorithms for CNF satisfiability have running times of the form $O(2^{(1-o(1))n} \cdot \mbox{poly}(m))$. }  While our results are conditional, SETH has been employed in many recent hardness results. Its plausibility stems from the fact that, despite 60 years of research on satisfiability algorithms, no such improvement has been discovered.

We emphasize that our ERM lower bounds focus on algorithms that compute the solution to high accuracy, e.g., with a $\log 1/\eps$ dependence on the approximation error $\eps$.
A (doubly) logarithmic dependence on $1 / \eps$ is generally seen as the ideal rate of convergence in optimization.
Algorithms with this property have also been studied extensively in the machine learning community, e.g., see the references in \cite{BS}.
However, approximate solutions to ERM problems can be sufficient for good generalization in learning tasks.
Indeed, stochastic gradient descent (SGD) is often advocated as an efficient learning algorithm despite its polynomial dependence on $1 / \eps$ in the optimization error \cite{shalev2007pegasos,BB07}.
Our results support this viewpoint since SGD sidesteps the quadratic time complexity of our lower bounds.

Furthermore, our results do not rule out algorithms that achieve a sub-quadratic running time for well-behaved instances, e.g., instances with low-dimensional structure.
Indeed, many such approaches have been investigated in the literature, for instance the Nystr\"om method or random features for kernel problems \cite{WS01,RR08}.
Our results offer an explanation for the wide variety of techniques.
The lower bounds are evidence that there is no ``silver bullet'' algorithm for solving the aforementioned ERM problems in sub-quadratic time, to high accuracy, and for all instances.

\paragraph{Paper outline.}
In the next section, we briefly introduce the relevant background from fine-grained complexity theory.
Section \ref{sec:contributions} then presents our formal results.
In Sections~\ref{s:svm}, \ref{s:kpca}, and \ref{s:krr} we outline the proofs for kernel SVM, kernel PCA, and kernel ridge regression, respectively.
Section~\ref{s:retraining} describes our proof for neural networks.
Due to space constraints, we defer the proof details to the appendix in the supplementary  material.

\section{Background}
\label{sec:background}

\paragraph{Fine-grained complexity} We obtain our conditional hardness results via reductions from two well-studied problems: {\em Orthogonal Vectors} and {\em Bichromatic Hamming Close Pair}. 

\begin{definition}[Orthogonal Vectors problem (OVP)] \label{OVP}
Given two sets $A=\{a_1, \ldots, a_n\}\subseteq \{0,1\}^d$ and $B=\{b_1, \ldots, b_n\}\subseteq \{0,1\}^d$ of $n$ binary vectors, decide if there exists a pair $a \in A$ and $b \in B$ such that $a^\transp b = 0$.
\end{definition}

For OVP, we can assume without loss of generality that all vectors in $B$ have the same number of $1$s.
This can be achieved by appending $d$ entries to every $b_i$ and setting the necessary number of them to $1$ and the rest to $0$.
We then append $d$ entries to every $a_i$ and set all of them to $0$.

\begin{definition}[Bichromatic Hamming Close Pair (BHCP) problem] \label{BHCP}
  Given two sets $A=\{a_1, \ldots, a_n\}\subseteq \{0,1\}^d$ and $B=\{b_1, \ldots, b_n\}\subseteq \{0,1\}^d$ of $n$ binary vectors and an integer $t \in \{2, \ldots, d\}$, decide if there exists a pair $a \in A$ and $b \in B$ such that the number of coordinates in which they differ is less than $t$ (formally, \mbox{$\Hamming(a,b):=||a-b||_1<t$}). If there is such a pair $(a, b)$, we call it a \emph{close pair}.
\end{definition}

It is known that both OVP and BHCP require almost quadratic time (i.e.,  $n^{2-o(1)}$) for any $d=\omega(\log n)$ assuming SETH~\cite{alman2015probabilistic}.\footnote{We use $\omega(g(n))$ to denote any function $f$ such that $\lim_{n \to \infty} f(n)/g(n)=\infty$. Similarly, we use $o(g(n))$ to denote any function $f$  such that $\lim_{n \to \infty} f(n)/g(n)=0$. Consequently, we will refer to functions of the form  $\omega(1)$ as {\em super-constant} and to $n^{\omega(1)}$ as {\em super-polynomial}.}
Furthermore, if we allow the sizes $|A|=n$ and $|B|=m$ to be different, both problems require $(nm)^{1-o(1)}$ time assuming SETH, as long as $m=n^\alpha$ for some constant $\alpha \in (0,1)$~\cite{bringmann2015quadratic}.
Our proofs will proceed by embedding OVP and BHCP instances into ERM problems.
Such a reduction then implies that the ERM problem requires almost quadratic time if the SETH is true.
If we could solve the ERM problem faster, we would also obtain a faster algorithm for the satisfiability problem.
 
\section{Our contributions}
\label{sec:contributions}
We now describe our main contributions, starting with conditional hardness results for ERM problems.

\subsection{Kernel ERM problems}
We provide hardness results for multiple kernel problems.
In the following, let $x_1, \ldots, x_n \in \R^d$ be the $n$ input vectors, where $d=\omega(\log n)$.	
We use $y_1, \ldots, y_n \in \R$ as $n$ labels or target values.
Finally, let $k(x,x')$ denote a kernel function and let $K \in \R^{n \times n}$ be the corresponding kernel matrix, defined as $K_{i,j}:=k(x_i,x_j)$ \cite{scholkopf2001learning}.
Concretely, we focus on the Gaussian kernel $k(x,x'):=\exp\left(-C\|x-x'\|_2^2\right)$ for some $C>0$.
We note that our results can be generalized to any kernel with exponential tail.
	
\subsubsection{Kernel SVM}
For simplicity, we present our result for hard-margin SVMs without bias terms.
This gives the following optimization problem.
 
\begin{definition}[Hard-margin SVM] \label{def_primal}
	A (primal) hard-margin SVM is  an optimization problem of the following form:
	\begin{equation} \label{eq_primal}
		\begin{aligned}
			& \underset{\alpha_1, \ldots, \alpha_n \geq 0}{\text{minimize}}
			& & {1 \over 2}\sum_{i,j=1}^n \alpha_i \, \alpha_j \, y_i \, y_j \, k(x_i,x_j) \\
			& \text{subject to}
			& & y_i f(x_i) \geq 1, \ \ i = 1, \ldots, n,
		\end{aligned}
	\end{equation}
	where $f(x):=\sum_{i=1}^n \alpha_i y_i k(x_i,x)$.
\end{definition}

The following theorem is our main result for SVMs.
In Section~\ref{s:hardness}  we provide similar hardness results for other common SVM variants.
\begin{theorem} \label{hardness}	
Let $k(a,a')$ be the Gaussian kernel with $C=100\log n$ and let $\eps=\exp(-\omega(\log^2 n))$.
Then approximating the optimal value of Equation \eqref{eq_primal} within a multiplicative factor $1+\eps$ requires almost quadratic time assuming SETH.
	\end{theorem}

\subsubsection{Kernel Ridge Regression}
Next we consider Kernel Ridge Regression, which is formally defined as follows.

\begin{definition}[Kernel ridge regression] \label{def_ridge}
	Given a real value $\lambda\geq 0$, the goal of kernel ridge regression is to output
	$$
		\argmin_{\alpha \in \R^n} {1 \over 2}||y-K\alpha||_2^2+{\lambda \over 2}\alpha^\transp K \alpha.
	$$
\end{definition}
This problem is equivalent to  computing  the vector $(K+\lambda I)^{-1}y$.
We focus on the special case where $\lambda=0$ and the vector $y$ has all equal entries $y_1=\ldots = y_n=1$.
In this case, the entrywise sum of $K^{-1}y$ is equal to the sum of the entries in $K^{-1}$.
Thus, it is sufficient to show hardness for computing the latter quantity. 

\begin{theorem} \label{ridge_regression}
Let $k(a,a')$ be the Gaussian kernel for any parameter $C=\omega(\log n)$ and let $\eps=\exp(-\omega(\log^2 n))$. 
Then computing the sum of the entries in $K^{-1}$ up to a multiplicative factor of $1+\eps$ requires almost quadratic time assuming SETH. 
\end{theorem}

\subsubsection{Kernel PCA}
Finally, we turn to the Kernel PCA problem, which we define as follows \cite{Murphy12}.
\begin{definition}[Kernel Principal Component Analysis (PCA)] \label{def_PCA}
	Let $1_n$ be an $n \times n$ matrix where each entry takes value $1/n$, and define $K':=(I-1_n)K(I-1_n)$. The goal of the kernel PCA problem is to output the $n$ eigenvalues of the matrix $K'$.
\end{definition}
In the above definition, the output only consists of the eigenvalues, not the eigenvectors. This is because computing all $n$ eigenvectors trivially takes at least quadratic time since the output itself has quadratic size.
Our hardness proof applies to the potentially simpler problem where only the eigenvalues are desired.
In fact, we show that even  computing {\em the sum} of the eigenvalues (i.e., the {\em trace} of the matrix) is hard.

\begin{theorem}
Let $k(a,a')$ be the Gaussian kernel with $C=100\log n$ and let $\eps=\exp(-\omega(\log^2 n))$.
Then approximating the sum of the eigenvalues of $K'=(I-1_n)K(I-1_n)$ within a multiplicative factor of $1+\eps$ requires almost quadratic time assuming SETH.	
\end{theorem}

\subsection{Neural network ERM problems}
\label{ss:retraining}
We now consider neural networks.
We focus on the problem of optimizing the top layer while keeping lower layers unchanged.
An instance of this problem is transfer learning with large networks that would take a long time and many examples to train from scratch \cite{RASC14}.
We consider neural networks of depth 2, with the sigmoid or ReLU activation function.
Our hardness result holds for a more general class of ``nice'' activation functions $S$ as described later (see Definition~\ref{activation}).

Given $n$ vectors $b_1, \ldots, b_n \in \R^d$ and $n$ weights $\alpha_1, \ldots, \alpha_n \in \R$, consider the function $f : \R^d \to \R$ using a non-linearity $S : \R \rightarrow \R$:
\[
	f(a)\; := \; \sum_{j=1}^n \alpha_j \cdot S(a^\transp b_j) \; .
\]
This function can be implemented as a neural net that has $d$ inputs, $n$ nonlinear activations (units), and one linear output. 

To complete the ERM problem, we also require a loss function.
Our hardness results hold for a large class of ``nice'' loss functions, which includes the hinge loss and the logistic loss.\footnote{In the binary setting we consider, the logistic loss is equivalent to the softmax loss commonly employed in deep learning.}

Given a nice loss function and $m$ input vectors $a_1, \ldots, a_m \in \R^d$ with corresponding labels $y_i$, we consider the following problem:
\begin{equation} \label{opt1}
	\begin{aligned}
		& \underset{\alpha_1, \ldots, \alpha_n \in \R}{\text{minimize}}
		& & \sum_{i=1}^m \loss(y_i,f(a_i)).
	\end{aligned}
\end{equation}

Our main result is captured by the following theorem. For simplicity, we set $m=n$.
\begin{theorem}
\label{t:retraining}
	For any $d=\omega(\log n)$, approximating the optimal value in Equation \eqref{opt1} up to a  multiplicative factor of $1+{1 \over 4n}$ requires almost quadratic time assuming SETH.
\end{theorem}

\subsection{Hardness of gradient computation}

Finally, we consider the problem of computing the gradient of the loss function for a given set of examples.
We focus on the network architecture from the previous section.
Formally, we obtain the following result:

\begin{theorem}
\label{t:gradient}
Consider the empirical risk in Equation \eqref{opt1} under the following assumptions:
(i) The function $f$ is represented by a neural network with $n$ units, $n \cdot d$ parameters, and the ReLU activation function.
(ii) We have  $d=\omega(\log n)$.
(iii) The loss function is the logistic loss or hinge loss.
Then computing a gradient of the empirical risk for $m$ examples takes at least $O((nm)^{1-o(1)})$ time assuming SETH.
\end{theorem}
In Section \ref{sec:gradient}, we also prove a similar statement for the sigmoid activation function. 

\subsection{Related work}
Recent work has demonstrated conditional quadratic hardness results for many optimization problems over graphs and sequences.
These results include computing diameter in sparse graphs~\cite{roditty2013fast, chechik2014better}, Local Alignment~\cite{abboud2014consequences}, Fr\'echet distance~\cite{bringmann2014walking}, Edit Distance~\cite{backurs2015edit}, Longest Common Subsequence, and Dynamic Time Warping~\cite{abboud2015tight,bringmann2015quadratic}.
As in our paper, the SETH and related assumptions underlie these lower bounds.
To the best of our knowledge, our paper is the first application of this methodology to continuous (as opposed to combinatorial) optimization problems. 

There is a long line of work on the oracle complexity of optimization problems, going back to \cite{NY83}.
We refer the reader to \cite{Nesterov} for these classical results.
The oracle complexity of ERM problems is still subject of active research, e.g., see \cite{AB15,cesa2015complexity,woodworth2016tight,AS16b,AS16a}.
The work closest to ours is \cite{cesa2015complexity}, which gives quadratic time lower bounds for ERM algorithms that access the kernel matrix through an evaluation oracle or a low-rank approximation.

The oracle results are fundamentally different from the lower bounds presented in our paper.
Oracle lower bounds are typically unconditional, but inherently apply only to a limited class of algorithms due to their information-theoretic nature.
Moreover, they do not account for the \emph{cost} of executing the oracle calls, as they merely lower bound their number.
In contrast, our results are conditional (based on the SETH and related assumptions), but apply to \emph{any} algorithm and account for the {\em total} computational cost. 
This significantly broadens the reach of our results.
We show that the hardness is not due to the oracle abstraction but instead inherent in the computational problem.
 
\section{Overview of the hardness proof for kernel SVMs}
\label{s:svm} \label{ss:overview}

Let $A=\{a_1, \ldots, a_n\}\subseteq \{0,1\}^d$ and $B=\{b_1, \ldots, b_n\}\subseteq \{0,1\}^d$ be the two sets of binary vectors from the BHCP instance with $d=\omega(\log n)$.
Our goal is to determine whether there is a close pair of vectors. We show how to solve this BHCP instance by performing {\em three} computations of SVM:
	\begin{enumerate}
		\item We take the first set $A$ of binary vectors, assign label $1$ to all vectors, and solve the corresponding SVM on the $n$ vectors:
			\begin{equation} \label{SVM1}
				\begin{aligned}
					& \underset{\alpha_1, \ldots, \alpha_n \geq 0}{\text{minimize}}
					& & {1 \over 2}\sum_{i,j=1}^n \alpha_i \alpha_j k(a_i,a_j) \\
					& \text{subject to}
					& & \sum_{j=1}^n \alpha_j k(a_i,a_j) \geq 1, \ \ i = 1, \ldots, n.
				\end{aligned}
			\end{equation}
		Note that we do not have $y_i$ in the expressions because all labels are $1$.
		\item We take the second set $B$ of binary vectors, assign label $-1$ to all vectors, and solve the corresponding SVM on the $n$ vectors:
			\begin{equation} \label{SVM2}
				\begin{aligned}
					& \underset{\beta_1, \ldots, \beta_n \geq 0}{\text{minimize}}
					& & {1 \over 2}\sum_{i,j=1}^n \beta_i \beta_j k(b_i,b_j) \\
					& \text{subject to}
					& & - \sum_{j=1}^n \beta_j k(b_i,b_j) \leq -1, \ \ i = 1, \ldots, n.
				\end{aligned}
			\end{equation}
		\item We take both sets $A$ and $B$ of binary vectors, assign label $1$ to all vectors from the first set $A$ and label $-1$ to all vectors from the second set $B$. We then solve the corresponding SVM on the $2n$ vectors:
			\begin{equation} \label{SVM3}
				\begin{aligned}
          \iftoggle{singlecolumn}{
					& \underset{\substack{\alpha_1, \ldots, \alpha_n \geq 0\\ \beta_1, \ldots, \beta_n \geq 0}}{\text{minimize}}
					& & {1 \over 2}\sum_{i,j=1}^n \alpha_i \alpha_j k(a_i,a_j) + {1 \over 2}\sum_{i,j=1}^n \beta_i \beta_j k(b_i,b_j) - \sum_{i,j=1}^n \alpha_i \beta_j k(a_i,b_j)\\
					& \text{subject to}
					& & \sum_{j=1}^n \alpha_j k(a_i,a_j) - \sum_{j=1}^n \beta_j k(a_i,b_j) \geq 1, \quad i = 1, \ldots, n \; ,\\
					&&& - \sum_{j=1}^n \beta_j k(b_i,b_j)  + \sum_{j=1}^n \alpha_j k(b_i,a_j)\leq -1, \quad i = 1, \ldots, n \;.
          }{
					& \underset{\substack{\alpha_1, \ldots, \alpha_n \geq 0\\ \beta_1, \ldots, \beta_n \geq 0}}{\text{minimize}}
					& & {1 \over 2}\sum_{i,j=1}^n \alpha_i \alpha_j k(a_i,a_j) + {1 \over 2}\sum_{i,j=1}^n \beta_i \beta_j k(b_i,b_j) \\ 
					&&& \hspace{3.0 cm} - \sum_{i,j=1}^n \alpha_i \beta_j k(a_i,b_j)\\
					& \text{subject to}
					& & \sum_{j=1}^n \alpha_j k(a_i,a_j) - \sum_{j=1}^n \beta_j k(a_i,b_j) \geq 1, \\
					&&& \hspace{4 cm} i = 1, \ldots, n,\\
					&&& - \sum_{j=1}^n \beta_j k(b_i,b_j)  + \sum_{j=1}^n \alpha_j k(b_i,a_j)\leq -1,\\ 
					&&& \hspace{4 cm} i = 1, \ldots, n.
          }
				\end{aligned}
			\end{equation}
	\end{enumerate}
	
	\paragraph{Intuition behind the construction.} To show a reduction from the BHCP problem to SVM computation, we have to consider two cases:
	\begin{itemize}
		\item The YES case of the BHCP problem when there are two vectors that are close in Hamming distance. That is, there exist $a_i \in A$ and $b_j \in B$ such that $\Hamming(a_i,b_j)<t$.
		\item The NO case of the BHCP problem when there is no close pair of vectors. That is, for all $a_i \in A$ and $b_j \in B$, we have $\Hamming(a_i,b_j)\geq t$.
	\end{itemize}

We show that we can distinguish between these two cases by comparing the objective value of the first two SVM instances above to the objective value of the third.

\paragraph{Intuition for the NO case.} We have $\Hamming(a_i,b_j)\geq t$ for all $a_i \in A$ and $b_j \in B$.
The Gaussian kernel then gives the inequality
	\begin{align*}
		k(a_i,b_j) & = \exp(-100\log n\cdot \|a_i-b_j\|_2^2)\\
		& \leq \exp(-100\log n \cdot t)
	\end{align*}
	for all $a_i \in A$ and $b_j \in B$.
  This means that the value $k(a_i,b_j)$ is very small. For simplicity, assume that it is equal to $0$, i.e., $k(a_i,b_j)=0$ for all $a_i \in A$ and $b_j \in B$.
	
 	Consider the third SVM \eqref{SVM3}.
  It contains three terms involving $k(a_i,b_j)$: the third term in the objective function, the second term in the inequalities of the first type, and the second term in the inequalities of the second type. We assumed that these terms are equal to $0$ and we observe that the rest of the third SVM is equal to the sum of the first SVM \eqref{SVM1} and the second SVM \eqref{SVM2}. Thus we expect that the optimal value of the third SVM is approximately equal to the sum of the optimal values of the first and the second SVMs. If we denote the optimal value of the first SVM \eqref{SVM1} by $\val(A)$, the optimal value of the second SVM \eqref{SVM2} by $\val(B)$, and the optimal value of the third SVM \eqref{SVM3} by $\val(A,B)$, then we can express our intuition in terms of the approximate equality
	$$
		\val(A,B)\approx\val(A)+\val(B) \; .
	$$
	
	\paragraph{Intuition for the YES case} In this case, there is a close pair of vectors $a_i \in A$ and $b_j \in B$ such that $\Hamming(a_i,b_j)\leq t-1$. Since we are using the Gaussian kernel we have the following inequality for this pair of vectors:
	\begin{align*}
		k(a_i,b_j) & = \exp(-100\log n\cdot \|a_i-b_j\|_2^2)\\
		& \geq \exp(-100\log n \cdot (t-1)) \; .
	\end{align*}
	
	We therefore have  a large summand in each of the three terms from the above discussion. Thus the three terms do not (approximately) disappear and there is no reason for us to expect that the approximate equality holds. We can thus expect
	$$
		\val(A,B)\not\approx\val(A)+\val(B) \; .
	$$
	
Thus, by computing $\val(A,B)$ and comparing it to $\val(A)+\val(B)$ we can distinguish between the YES and NO instances of BHCP. This completes the reduction. 
The full proofs are given in Section~\ref{s:hardness}.
 
\section{Hardness proof for Kernel PCA}
\label{s:kpca}
In this section, we present the full proof of quadratic hardness for Kernel PCA. It will also be helpful for Kernel Ridge Regression in the next section.

Given a matrix $X$, we denote its trace (the sum of the diagonal entries) by $\tr(X)$ and  the total sum of its entries by $s(X)$. In the context of the matrix $K'$ defining our problem, we have the following equality:
\begin{align*}
	\tr(K')& =\tr((I-1_n)K(I-1_n))\\
	& =\tr(K(I-1_n)^2)=\tr(K(I-1_n))\\
	& =\tr(K)-\tr(K1_n)=n-s(K)/n \; .
\end{align*}
Since the sum of the eigenvalues is equal to the trace of the matrix and $\tr(K')=n-s(K)/n$,
it is sufficient to show hardness for computing $s(K)$. The following lemma completes the proof of the theorem.

\begin{lemma} \label{sum_hardness}
	Computing $s(K)$ within multiplicative error $1+\eps$ for $\eps=\exp(-\omega(\log^2n))$ requires almost quadratic time assuming SETH. 
\end{lemma}
\begin{proof}
	As for SVMs, we will reduce the BHCP problem to the computation of $s(K)$. Let $A$ and $B$ be the two sets of $n$ binary vectors coming from an instance of the BHCP problem. Let $K_A, K_B \in \R^{n \times n}$ be the kernel matrices corresponding to the sets $A$ and $B$, respectively. Let $K_{A,B} \in \R^{2n \times 2n}$ be the kernel matrix corresponding to the set $A \cup B$. We observe that
	\begin{align*}
		s:=& (s(K_{A,B})-s(K_A)-s(K_B))/2\\
		=& \sum_{i,j=1}^n k(a_i,b_j)\\
		=& \sum_{i,j=1}^n \exp(-C||a_i-b_j||_2^2).
	\end{align*}
	
	Now we consider two cases.
	\paragraph{Case 1.} There are no close pairs, that is, for all $i,j=1, \ldots, n$ we have $||a_i-b_j||_2^2\geq t$ and $\exp(-C||a_i-b_j||_2^2)\leq \exp(-Ct)=:\delta$. Then $s\leq n^2\delta$.
	
	\paragraph{Case 2.} There is a close pair. That is, $||a_{i'}-b_{j'}||_2^2\leq t-1$ for some $i',j'$. This implies that $\exp(-C||a_{i'}-b_{j'}||_2^2)\geq \exp(-C(t-1))=:\Delta$. Thus, $s\geq \Delta$.
	
	Since $C=100\log n$, we have that $\Delta\geq n^{10}\delta$ and we can distinguish the two cases.
	
	\paragraph{Precision.} To distinguish $s\geq \Delta$ from $s\leq n^2 \delta$, it is sufficient that $\Delta\geq 2n^2\delta$. This holds for $C=100\log n$. The sufficient additive precision is $\exp(-Cd)=\exp(-\omega(\log^2n))$. Since $s(K)\leq O(n^{2})$ for any Gaussian kernel matrix $K$, we also get that $(1+\eps)$ multiplicative approximation is sufficient to distinguish the cases for any $\eps=\exp(-\omega(\log^2n))$.
\end{proof}
 
\section{Overview of the hardness proof for Kernel Ridge Regression}
\label{s:krr}
The hardness proof for Kernel Ridge Regression is somewhat technical, so we only outline it here.
As described in Section \ref{sec:contributions}, our goal is to prove that computing the sum of the entries in the matrix $K^{-1}$ requires quadratic time.
The main idea is to leverage the hardness of computing the sum $\sum_{i,j=1}^n k(a_i,b_j)$, as shown in Lemma~\ref{sum_hardness}. 

To this end, we show that the following approximate equality holds:
	\begin{align*}
		(s(K_A^{-1})+s(K_B^{-1})-s(K_{A,B}^{-1}))/2
		\approx \sum_{i,j=1}^n k(a_i,b_j).
	\end{align*}
where we use the notation $K_{A,B}, K_A, K_B$ as in the previous section.  This allows us to conclude that a sufficiently accurate approximation to $s(K^{-1})$ for a kernel matrix $K$ gives us the solution to the BHCP problem.
The key observation for proving the above approximate equality is the following: if a matrix has no large off-diagonal element, i.e., it is an ``almost identity'', then so is its inverse. 
 
\section{Overview of the hardness proof for training the final layer of a neural network}
\label{s:retraining}

We start by formally defining the class of ``nice'' loss functions and ``nice'' activation functions. 

\begin{definition} \label{nice}
	For a label $y \in \{-1,1\}$ and a prediction $w \in \R$, we call the loss function $\loss(y,w):\{-1,1\} \times \R \to \R_{\geq 0}$ \emph{nice} if the following three properties hold:
	\begin{itemize}
		\item $\loss(y,w)=l(yw)$ for some convex function $l:\R \to \R_{\geq 0}$.
    \item For some sufficiently large constant $K>0$, we have that (i) $l(x)\leq o(1)$ for all $x\geq n^K$, (ii) $l(x)\geq \omega(n)$ for all $x \leq -n^K$, and (iii) $l(x) = l(0) \pm o(1/n)$ for all $x \in \pm O(n^{-K})$.
		\item $l(0)>0$ is some constant strictly larger than $0$.
	\end{itemize}
\end{definition}

We note that the hinge loss function $\loss(y,x)=\max(0,1-y \cdot x)$ and the logistic loss function $\loss(y,x)={1 \over \ln 2}\ln\left(1+e^{-y\cdot x}\right)$ are nice loss functions according to the above definition.

\begin{definition} \label{activation}
	A  \emph{non-decreasing} activation functions $S:\R\to\R_{\geq 0}$ is ``nice'' if it satisfies the following property:
	for all sufficiently large constants $T>0$ there exist $v_0>v_1>v_2$ such that
	\begin{itemize}
		\item $S(v_0)=\Theta(1)$;
		\item $S(v_1)=1/n^T$;
		\item $S(v_2)=1/n^{\omega(1)}$;
		\item $v_1=(v_0+v_2)/2$.
	\end{itemize}
\end{definition}

The ReLU activation $S(z)=\max(0,z)$ satisfies these properties since we can choose $v_0=1$, $v_1=1/n^T$, and $v_2=-1+2/n^T$.
For the sigmoid function $S(z)={1 \over 1+e^{-z}}$, we can choose $v_1=-\log(n^T-1)$, $v_0=v_1+C$, and $v_2=v_1-C$ for some $C=\omega(\log n)$. In the rest of the proof we set $T=1000K$, where $K$ is the constant from Definition~\ref{nice}.

We now describe the proof of Theorem~\ref{t:retraining}. 
We use the notation $\alpha:=(\alpha_1, \ldots, \alpha_n)^T$.
Invoking the first property from Definition \ref{nice}, we observe that the optimization problem \eqref{opt1} is equivalent to the following optimization problem:
\begin{equation} \label{optM}
	\begin{aligned}
		& \underset{\alpha \in \R^n}{\text{minimize}}
		& & \sum_{i=1}^m l(y_i \cdot (M\alpha)_i),
	\end{aligned}
\end{equation}
where $M \in \R^{m \times n}$ is the matrix defined as $M_{i,j}:=S(a_i^\transp b_j)$ for $i=1, \ldots, m$ and $j=1, \ldots n$.
For the rest of the section we will use $m=\Theta(n)$.

Let $A=\{a_1, \ldots, a_n\} \subseteq \{0,1\}^d$ and $B=\{b_1, \ldots, b_n\} \subseteq \{0,1\}^d$ with $d=\omega(\log n)$ be the input to the Orthogonal Vectors problem. To show hardness we construct matrix a $M$ as combination of $3$ smaller matrices:
$$
	M=
		\left[
			\begin{array}{c}
				M_1 \\
				M_2 \\
				M_2
			\end{array}
		\right].
$$
Both matrices $M_1, M_2 \in \R^{n \times n}$ are of size $n \times n$.
Thus the number of rows of $M$ is $m=3n$. 

We select the input examples and weights so that the constructed matrices have the following properties:
\begin{itemize}
  \item $M_1$: if two vectors $a_i$ and $b_j$ are orthogonal, then the corresponding entry $(M_1)_{i,j}=S(v_0)=\Theta(1)$ and otherwise $(M_1)_{i,j}\approx 0$.\footnote{We write $x \approx y$ if $x=y$ up to an inversely superpolynomial additive factor, i.e., $|x-y|\leq n^{-\omega(1)}$.}
\item $M_2$: $(M_2)_{i,i}=S(v_1)=1/n^{1000K}$ and $(M_2)_{i,j}  \approx 0$ for all $i \neq j$ 
\end{itemize}

To complete the description of the optimization problem \eqref{optM}, we assign labels to the inputs corresponding to the rows of the matrix $M$.
We assign label $1$ to all inputs corresponding to rows of the matrix $M_1$ and the first copy of the matrix $M_2$.
We assign label $-1$ to all remaining rows of the matrix $M$ corresponding to the second copy of matrix $M_2$.
	
The proof of the theorem is completed by the following two lemmas.
	
	\begin{lemma} \label{upper_bound}
		If there is a pair of orthogonal vectors, then the optimal value of \eqref{optM} is upper bounded by $(3n-1)\cdot l(0)+o(1)$.
	\end{lemma}
	
	\begin{lemma} \label{lower_bound}
		If there is no pair of orthogonal vectors, then the optimal value of \eqref{optM} is lower bounded by $3n\cdot l(0)-o(1)$.
	\end{lemma}
	
	The intuition behind the proofs of the two lemmas is as follows.
  To show Lemma~\ref{upper_bound}, it suffices to demonstrate a feasible vector $\alpha$. We achieve this by setting each entry of $\alpha$ to a ``moderately large'' number. In this way we ensure that:
	\begin{itemize}
	\item At most $n-1$ entries of $M_1 \alpha$ are close to $0$, each contributing at most $l(0)+o(1/n)$ to the loss.
	\item The remaining entries of $M_1 \alpha$ are positive and ``large'', each contributing only $o(1)$ to the loss. These entries correspond to orthogonal pairs of vectors. The total loss corresponding to $M_1$ is thus upper bounded by $(n-1)\cdot l(0)+o(1)$.
	\item The entries of the two vectors $M_2 \alpha$ are close to $0$, contributing approximately $2n \cdot l(0)$ to the loss.
	\end{itemize}
	
	To show Lemma~\ref{lower_bound}, we need to argue that no vector $\alpha$ achieves loss smaller than $3n\cdot l(0)-o(1)$.
  We first observe that the two copies of $M_2$ contributes a loss of at least $2n\cdot l(0)$, independently of the values of $\alpha$.
  We then  show that the entries of $\alpha$ cannot be too large (in absolute value), as otherwise the loss incurred by the two copies of the matrix $M_2$ would be too large.
  Using this property, we show that the loss incurred by $M_1$ is lower bounded by $n\cdot l(0)-o(1)$.  
	
	Since $l(0)>0$ (the third property in Definition \ref{nice}), we can distinguish $\leq 3n \cdot l(0) + o(1)$ from $\geq (3n + 1) \cdot l(0) - o(1)$ with a high-accuracy ERM solution. Hence Lemmas \ref{upper_bound} and \ref{lower_bound} imply that Problem \eqref{optM} is SETH hard.

\section{Hardness proof for gradient computation}
\label{sec:gradient}

Finally, we consider the problem of computing the gradient of the loss function for a given set of examples. We focus on the network architecture as in the previous section. We first prove the following lemma. 

\begin{lemma} \label{gradient}
	Let $F_{\alpha, B}(a):=\sum_{j=1}^n \alpha_j S(a, b_j)$ be the output of a neural net with activation function $S$, where
	\begin{itemize}
		\item $a$ is an input vector from the set $A := \{a_1, \ldots, a_m\} \subseteq \{0,1\}^d$;
		\item $B:=\{b_1, \ldots, b_n\} \subseteq \{0,1\}^d$ is a set of binary vectors;
		\item $\alpha=\{\alpha_1, \ldots, \alpha_n\}^T \in \R^n$ is an $n$-dimensional real-valued vector.
	\end{itemize}
	For some loss function $l:\R\to\R$, let $l(F_{\alpha, B}(a))$ be the loss for input $a$ when the label of the input $a$ is $+1$.
	
	Consider the gradient of the total loss $l_{\alpha,A,B}:=\sum_{a \in A}l(F_{\alpha,B}(a))$ at $\alpha_1=\ldots=\alpha_n=0$ with respect to $\alpha_1, \ldots, \alpha_n$.
	The sum of the entries of the gradient is equal to $l'(0)\cdot \sum_{a\in A, b \in B}S(a,b)$, where $l'(0)$ is the derivative of the loss function $l$ at $0$.
\end{lemma}

\begin{proof}
	The proof follows from
	\begin{align*}
		{\partial l_{\alpha,A,B} \over \partial \alpha_j}& =\sum_{a \in A} {\partial l(F_{\alpha,B}(a))\over \partial F_{\alpha,B}(a)}S(a,b_j)\\
		& =l'(0)\cdot \sum_{a \in A}S(a,b_j) \hspace{.7 cm} \text{(since }F_{\alpha,B}(a)=0\text{)}.
	\end{align*}
\end{proof}

For the hinge loss function, we have that the loss function is $l(x)=\max(0,1-x)$ if the label is $+1$. Thus, $l'(0)=-1$.
For the logistic loss function, we have that the loss function is $l(x)={1 \over \ln 2}\ln\left(1+e^{-x}\right)$ if the label is $+1$. Thus, $l'(0)=-{1 \over 2\ln 2}$ in this case.

\begin{proof}[Proof of Theorem~\ref{t:gradient}]
	We set $S(a,b):=\max(0,1-2a^\transp b)$. 
	Using Lemma \ref{gradient}, we get that the sum of entries of the gradient of the total loss function is equal to $l'(0)\cdot\sum_{a\in A, b \in B}1_{a^\transp b=0}$.
  Since $l'(0)\neq 0$, this reduces OV to the gradient computation problem. 
\end{proof}

We can also show the same statement holds for the sigmoid activation function. 

\begin{theorem}
	Consider a neural net with of size $n$ with the sigmoid activation function $\sigma(x)={1 \over 1+e^{-x}}$.
  Computing a gradient of the empirical loss for $m$ examples requires time $(nm)^{1-o(1)}$ assuming SETH.
\end{theorem}
\begin{proof}
	We set $S(a,b):=\sigma(-10 (\log n) \cdot a^\transp b)$. Using Lemma \ref{gradient}, we get that the sum of entries of the gradient is equal to $l'(0)\cdot\sum_{a\in A, b \in B}{1 \over 1+e^{10 (\log n) \cdot a^\transp b}}$.
  It is easy to show that this quantity is at least $l'(0)/2$ if there is an orthogonal pair and at most $l'(0)/n$ otherwise.
  Since $l'(0)\neq 0$, we get the required hardness.
\end{proof}
 
\section{Conclusions}
We have shown that a range of kernel problems require quadratic time for obtaining a high accuracy solution unless the Strong Exponential Time Hypothesis is false.
These problems include variants of the kernel SVM, kernel ridge regression, and kernel PCA.
We also gave a similar hardness result for training the final layer of a depth-2 neural network.
This result is general and applies to multiple loss and activation functions.
Finally, we proved that computing the empirical loss gradient for such networks takes time that is essentially ``rectangular'', i.e., proportional to the product of the network size and the number of examples.

Several of our results are obtained by a reduction from the (exact) Bichromatic Hamming Closest Pair problem or the Orthogonal Vectors problem. This demonstrates a strong connection between kernel methods and similarity search, and suggests that perhaps a reverse reduction is also possible. Such a reduction could potentially lead to faster approximate algorithms for kernel methods: although the exact closest pair problem has no known sub-quadratic solution, efficient and practical sub-quadratic time algorithms for the approximate version of the problem exist (see e.g.,~\cite{andoni2006near,valiant2012finding,andoni2015optimal,AILRS15,alman2016polynomial}). 

Another interesting goal would be obtaining fine-grained lower bounds for other optimization problems.
 	
\bibliographystyle{alpha}
\newcommand{\etalchar}[1]{$^{#1}$}

\appendix

\section{Preliminaries}
\label{s:prelim}

In this section we define several notions used later in the paper. We start from the soft-margin support vector machine (see \cite{muller2001introduction}).

\begin{definition}[Support Vector Machine (SVM)] \label{def_soft}
	Let $x_1, \ldots, x_n \in \R^d$ be $n$ vectors and $y_1, \ldots, y_n \in \{-1,1\}$ be $n$ labels. 
	Let $k(x,x')$ be a kernel function.
	An optimization problem of the following form is a (primal) SVM.
	\begin{equation} \label{eq_soft_primal}
		\begin{aligned}
			& \underset{\substack{\alpha_1, \ldots, \alpha_n \geq 0, \ \ b \\ \xi_1, \ldots, \xi_n \geq 0}}{\text{minimize}}
			& & {\lambda \over 2}\sum_{i,j=1}^n \alpha_i \alpha_j y_i y_j k(x_i,x_j) \ + \ {1 \over n}\sum_{i=1}^n \xi_i\\
			& \text{subject to}
			& & y_i f(x_i) \geq 1-\xi_i, \ \ i = 1, \ldots, n,
		\end{aligned}
	\end{equation}
	where $f(x):=b+\sum_{i=1}^n \alpha_i y_i k(x_i,x)$ and $\lambda \geq 0$ is called the regularization parameter.
	$\xi_i$ are known as the slack variables.
	
	The dual SVM is defined as
	\begin{equation} \label{eq_soft_dual}
		\begin{aligned}
			& \underset{\alpha_1, \ldots, \alpha_n \geq 0}{\text{maximize}}
			& & \sum_{i=1}^n \alpha_i - {1 \over 2}\sum_{i,j=1}^n \alpha_i \alpha_j y_i y_j k(x_i,x_j) \\
			& \text{subject to}
			& & \sum_{i=1}^n \alpha_i y_i=0, \\
			& & & \alpha_1, \ldots, \alpha_n \leq {1 \over \lambda n}.
		\end{aligned}
	\end{equation}
	
	We refer to the quantity $b$ as the bias term.
	When we require that the bias is $b=0$, we call the optimization problem as SVM without the bias term.
	The primal SVM without the bias term remains the same except $f(x)=\sum_{i=1}^n \alpha_i y_i k(x_i,x)$. The dual SVM remains the same except we remove the equality constraint $\sum_{i=1}^n \alpha_i y_i=0$.
\end{definition}

 The (primal) hard-margin SVM defined in the previous section corresponds to soft-margin SVM in the setting when $\lambda \to 0$.
The dual hard-margin SVM is defined as follows.

\begin{definition}[Dual hard-margin SVM] \label{def_dual}
	Let $x_1, \ldots, x_n \in \R^d$ be $n$ vectors and $y_1, \ldots, y_n \in \{-1,1\}$ be $n$ labels. 
	Let $k(x,x')$ be a kernel function.
	An optimization problem of the following form is a dual hard-margin SVM.
	\begin{equation} \label{eq_dual}
		\begin{aligned}
			& \underset{\alpha_1, \ldots, \alpha_n \geq 0}{\text{maximize}}
			& & \sum_{i=1}^n \alpha_i - {1 \over 2}\sum_{i,j=1}^n \alpha_i \alpha_j y_i y_j k(x_i,x_j) \\
			& \text{subject to}
			& & \sum_{i=1}^n \alpha_i y_i=0.
		\end{aligned}
	\end{equation}
	
	If the primal hard-margin SVM is without the bias term ($b=0$), then we omit the inequality constraint $\sum_{i=1}^n \alpha_i y_i=0$ in the dual SVM.
\end{definition}

We will use the following fact (see \cite{muller2001introduction}).
\begin{fact}\label{primal_dual}
	If $\alpha_1^*, \ldots, \alpha_n^*$ achieve the minimum in an SVM, then the same $\alpha_1^*, \ldots, \alpha_n^*$ achieve the maximum in the dual SVM. Also, the minimum value and the maximum value are equal. 
\end{fact}
 
\section{Hardness for SVM without the bias term term}
\label{s:hardness}
	In this section we formalize the intuition from Section~\ref{ss:overview}. We start from the following two lemmas.

	\begin{lemma}[NO case] \label{nocase}
		If for all $a_i \in A$ and $b_j \in B$ we have $\Hamming(a_i,b_j) \geq t$, then
		$$
			\val(A,B)\leq \val(A)+\val(B) + 200n^6 \exp(-100\log n \cdot t).
		$$
	\end{lemma}
	
	\begin{lemma}[YES case] \label{yescase}
		If there exist $a_i \in A$ and $b_j \in B$ such that $\Hamming(a_i,b_j) \leq t-1$, then
		$$
			\val(A,B)\geq \val(A)+\val(B) + {1 \over 4} \exp(-100\log n \cdot (t-1)).
		$$
	\end{lemma}
	
	Assuming the two lemmas we can distinguish the NO case from the YES case because 
	$$
		200n^6 \exp(-100\log n \cdot t) \ll {1 \over 4} \exp(-100\log n \cdot (t-1))
	$$
	by our choice of the parameter $C=100 \log n$ for the Gaussian kernel.
	
	Before we proceed with the proofs of the two lemmas, we prove the following auxiliary statement.
	
	\begin{lemma}\label{bounded}
		Consider SVM \eqref{SVM1}. Let $\alpha_1^*, \ldots, \alpha_n^*$ be the setting of values for $\alpha_1, \ldots, \alpha_n$ that achieves $\val(A)$. Then for all $i=1, \ldots, n$ we have that $n \geq \alpha_i^* \geq 1/2$.
		
		Analogous statement holds for SVM \eqref{SVM2}.
	\end{lemma}
	\begin{proof}
		First we note that $\val(A)\leq n^2/2$ because the objective value of \eqref{SVM1} is at most $n^2/2$ if we set $\alpha_1=\ldots=\alpha_n=1$. Note that all inequalities of \eqref{SVM1} are satisfied for this setting of variables. Now we lower bound $\val(A)$:
		$$
			\val(A) = {1 \over 2}\sum_{i,j}^n \alpha_i^* \alpha_j^* k(a_i,a_j) \geq {1 \over 2}\sum_{i=1}^n (\alpha_i^*)^2.
		$$
		From $\val(A) \geq {1 \over 2}\sum_{i=1}^n (\alpha_i^*)^2$ and $\val(A) \leq n^2/2$ we conclude that $\alpha_i^*\leq n$ for all $i$.
		
		Now we will show that $\alpha_i^* \geq 1/2$ for all $i=1, \ldots, n$.
		Consider the inequality 
		$$
			\sum_{j=1}^n \alpha_j^* k(a_i,a_j)=\alpha_i^* + \sum_{j \ : \ j \neq i} \alpha_j^* k(a_i,a_j) \geq 1
		$$ which is satisfied by $\alpha_1^*, \ldots, \alpha_n^*$ because this is an inequality constraint in \eqref{SVM1}. Note that $k(a_i,a_j) \leq {1 \over 10n^2}$ for all $i \neq j$ because $C=100 \log n$ and $\|a_i-a_{j}\|_2^2=\Hamming(a_i,a_j)\geq 1$ for all $i \neq j$. Also, we already obtained that $\alpha_j^* \leq n$ for all $j$. This gives us the required lower bound for $\alpha_i^*$:
		$$
			\alpha_i^*\geq 1-\sum_{j \ : \ j \neq i} \alpha_j^* k(a_i,a_j) \geq 1-n \cdot n \cdot {1 \over 10n^2} \geq 1/2.
		$$
	\end{proof}
	
	\paragraph{Additive precision} For particular value of $t$, the sufficient additive precision for solving the three SVMs is ${1 \over 100} \exp(-100\log n \cdot (t-1))$ to be able to distinguish the NO case from the YES case. Since we want to be able to distinguish the two cases for any $t \in \{2, \ldots, d\}$, it suffices to have an additive precision $\exp(-100\log n \cdot d) \leq {1 \over 100} \exp(-100\log n \cdot (t-1))$. From~\cite{alman2015probabilistic} we know that any $d=\omega(\log n)$ is sufficient to show hardness. Therefore, any additive approximation $\exp(-\omega(\log^2 n))$ is sufficient to show the hardness for SVM.
	
	\paragraph{Multiplicative precision} Consider any $\eps =\exp(-\omega(\log^2 n))$ and suppose we can approximate within multiplicative factor $(1+\eps)$ quantities $\val(A)$, $\val(B)$ and $\val(A,B)$. From the proof of Lemma \ref{bounded} we know that $\val(A), \val(B)\leq n^2/2$. If $\val(A,B)\leq 10n^2$, then $(1+\eps)$-approximation of the three quantities allows us to compute the three quantities within additive $\exp(-\omega(\log^2 n))$ factor and the hardness follows from the previous paragraph. On the other hand, if $\val(A,B)>10n^2$, then $(1+\eps)$-approximation of $\val(A,B)$ allows us to determine that we are in the YES case.
	
	In the rest of the section we complete the proof of the theorem by proving Lemma \ref{nocase} and Lemma \ref{yescase}.
	
	\begin{proof}[Proof of Lemma \ref{nocase}]
		Let $\alpha_1^*, \ldots, \alpha_n^*$ and $\beta_1^*, \ldots, \beta_n^*$ be the optimal assignments to SVMs \eqref{SVM1} and \eqref{SVM2}, respectively. We use the notation $\delta:=\exp(-100 \log n \cdot t)$. Note that $k(a_i,b_j)=\exp(-100\log n \cdot \|a_i-b_j\|_2^2)\leq \delta$ for all $i, j$ because $||a_i-b_{j}||_2^2=\Hamming(a_i,b_j)\geq t$ for all $i,j$.
		
		We define $\alpha_i':=\alpha_i^*+10n^{2}\delta$ and $\beta_i':=\beta_i^*+10n^{2}\delta$ for all $i=1, \ldots, n$.
		We observe that $\alpha_i',\beta_i'\leq 2n$ for all $i$ because $\alpha_i^*,\beta_i^*\leq n$ for all $i$ (Lemma \ref{bounded}) and $\delta=\exp(-100\log n \cdot t)\leq {1 \over 10n^{2}}$. Let $V$ be the value of the objective function in \eqref{SVM3} when evaluated on $\alpha_i'$ and $\beta_i'$.
		
		We make two claims. We claim that $\alpha_i'$ and $\beta_i'$ satisfy the inequality constraints in \eqref{SVM3}. This implies that $\val(A,B)\leq V$ since \eqref{SVM3} is a minimization problem. Our second claim is that $V\leq \val(A)+\val(B)+200n^{6}\delta$. The two claims combined complete the proof of the lemma.
		
		We start with the proof of the second claim. We want to show that $V\leq \val(A)+\val(B)+200n^{6}\delta$. We get the following inequality:
		\begin{align*}
			V=& \ {1 \over 2}\sum_{i,j=1}^n \alpha_i' \alpha_j' k(a_i,a_j) \ + \ {1 \over 2}\sum_{i,j=1}^n \beta_i' \beta_j' k(b_i,b_j) \ - \ \sum_{i,j=1}^n \alpha_i' \beta_j' k(a_i,b_j)\\
			\leq & \ {1 \over 2}\sum_{i,j=1}^n \alpha_i' \alpha_j' k(a_i,a_j) \ + \ {1 \over 2}\sum_{i,j=1}^n \beta_i' \beta_j' k(b_i,b_j)
		\end{align*}
		since the third sum is non-negative. It is sufficient to show two inequalities ${1 \over 2}\sum_{i,j=1}^n \alpha_i' \alpha_j' k(a_i,a_j) \leq \val(A)+100n^{6}\delta$ and ${1 \over 2}\sum_{i,j=1}^n \beta_i' \beta_j' k(b_i,b_j) \leq \val(B)+100n^{6}\delta$ to establish the inequality $V\leq \val(A)+\val(B)+200n^{6}\delta$. We prove the first inequality. The proof for the second inequality is analogous. We use the definition of $\alpha_i'=\alpha_i^*+10n^{2}\delta$:
		\begin{align*}
			& \ {1 \over 2}\sum_{i,j=1}^n \alpha_i' \alpha_j' k(a_i,a_j)\\
			=& \ {1 \over 2}\sum_{i,j=1}^n (\alpha_i^*+10n^{2}\delta)(\alpha_j^*+10n^{2}\delta)k(a_i,a_j)\\
			\leq & \ {1 \over 2}\sum_{i,j=1}^n \left(\alpha_i^* \alpha_j^* k(a_i,a_j)+20n^{3}\delta+100n^{4}\delta^2\right) \\
			\leq & \ \val(A)+100n^{6}\delta,
		\end{align*}
		where in the first inequality we use that $\alpha_i^*\leq n$ and $k(a_i,a_j)\leq 1$.
		
		Now we prove the first claim. We show that the inequality constraints are satisfied by $\alpha_i'$ and $\beta_i'$. We prove that the inequality 
		\begin{equation} \label{constraint}
			\sum_{j=1}^n \alpha_j' k(a_i,a_j) - \sum_{j=1}^n \beta_j' k(a_i,b_j) \geq 1
		\end{equation} 
		is satisfied for all $i=1, \ldots, n$. The proof that the inequalities $- \sum_{j=1}^n \beta_j' k(b_i,b_j) + \sum_{j=1}^n \alpha_j' k(b_i,a_j)\leq -1$ are satisfied is analogous. 
		
		We lower bound the first sum of the left hand side of \eqref{constraint} by repeatedly using the definition of $\alpha_i'=\alpha_i^*+10n^{2}\delta$:
		\begin{align*}
			& \ \sum_{j=1}^n \alpha_j'k(a_i,a_j)\\
			= & \ (\alpha_i^*	+10n^{2}\delta)+\sum_{j \ : \ j \neq i} \alpha_j'k(a_i,a_j) \\
			\geq & \ 10n^{2}\delta+\alpha_i^*+\sum_{j \ : \ j \neq i} \alpha_j^*k(a_i,a_j)\\
			= & \ 10n^{2}\delta+\sum_{j=1}^n \alpha_j^*k(a_i,a_j)\\
			\geq & \ 1+10n^{2}\delta.
		\end{align*}
		In the last inequality we used the fact that $\alpha_i^*$ satisfy the inequality constraints of SVM \eqref{SVM1}.
		
		We upper bound the second sum of the left hand side of \eqref{constraint} by using the inequality $\beta_j'\leq 2n$ and $k(a_i,b_j)\leq \delta$ for all $i,j$:
		$$
			\sum_{j=1}^n \beta_j' k(a_i,b_j) \leq 2n^2\delta.
		$$
		
		Finally, we can show that the inequality constraint is satisfied:
		$$
			\sum_{j=1}^n \alpha_j' k(a_i,a_j) - \sum_{j=1}^n \beta_j' k(a_i,b_j) \geq 1+10n^{2}\delta-2n^2\delta \geq 1.
		$$
	\end{proof}
	
	\begin{proof}[Proof of Lemma \ref{yescase}]
		To analyze the YES case, we consider the dual SVMs (see Definition \ref{def_dual}) of the three SVMs \eqref{SVM1}, \eqref{SVM2} and \eqref{SVM3}:
		\begin{enumerate}
			\item The dual SVM of SVM \eqref{SVM1}:
				\begin{equation} \label{dual_SVM1}
					\begin{aligned}
						& \underset{\alpha_1, \ldots, \alpha_n \geq 0}{\text{maximize}}
						& & \sum_{i=1}^n \alpha_i - {1 \over 2}\sum_{i,j=1}^n \alpha_i \alpha_j k(a_i,a_j).
					\end{aligned}
				\end{equation}
			\item The dual SVM of SVM \eqref{SVM2}:
				\begin{equation} \label{dual_SVM2}
					\begin{aligned}
						& \underset{\beta_1, \ldots, \beta_n \geq 0}{\text{maximize}}
						& & \sum_{i=1}^n \beta_i - {1 \over 2}\sum_{i,j=1}^n \beta_i \beta_j k(a_i,a_j).
					\end{aligned}
				\end{equation}
			\item The dual SVM of SVM \eqref{SVM3}:
				\begin{equation} \label{dual_SVM3}
					\begin{aligned}
						& \underset{\substack{\alpha_1, \ldots, \alpha_n \geq 0\\ \beta_1, \ldots, \beta_n \geq 0}}{\text{maximize}}
						& & \sum_{i=1}^n \alpha_i \ + \ \sum_{i=1}^n \beta_i \ - \ {1 \over 2}\sum_{i,j=1}^n \alpha_i \alpha_j k(a_i,a_j) \ - \ {1 \over 2}\sum_{i,j=1}^n \beta_i \beta_j k(b_i,b_j) \ + \ \sum_{i,j=1}^n \alpha_i \beta_j k(a_i,b_j).
					\end{aligned}
				\end{equation}
		\end{enumerate}
		
		Since the optimal values of the primal and dual SVMs are equal, we have that $\val(A)$, $\val(B)$ and $\val(A,B)$ are equal to optimal values of dual SVMs \eqref{dual_SVM1}, \eqref{dual_SVM2} and \eqref{dual_SVM3}, respectively (see Fact \ref{primal_dual}).
		
		Let $\alpha_1^*, \ldots, \alpha_n^*$ and $\beta_1^*, \ldots, \beta_n^*$ be the optimal assignments to dual SVMs \eqref{dual_SVM1} and \eqref{dual_SVM2}, respectively.
		
		Our goal is to lower bound $\val(A,B)$. Since \eqref{dual_SVM3} is a maximization problem, it is sufficient to show an assignment to $\alpha_i$ and $\beta_j$ that gives a large value to the objective function. For this we set $\alpha_i=\alpha_i^*$ and $\beta_j=\beta_j^*$ for all $i,j=1, \ldots, n$. This gives the following inequality:
		\begin{align*}
			\val(A,B)\geq & \ \sum_{i=1}^n \alpha_i^* \ + \ \sum_{i=1}^n \beta_i^* \ - \ {1 \over 2}\sum_{i,j=1}^n \alpha_i^* \alpha_j^* k(a_i,a_j) \ - \ {1 \over 2}\sum_{i,j=1}^n \beta_i^* \beta_j^* k(b_i,b_j) \ + \ \sum_{i,j=1}^n \alpha_i^* \beta_j^* k(a_i,b_j) \\
			\geq & \ \val(A) + \val(B) + \sum_{i,j=1}^n \alpha_i^* \beta_j^* k(a_i,b_j),
		\end{align*}
		where we use the fact that $\val(A)$ and $\val(B)$ are the optimal values of dual SVMs \eqref{dual_SVM1} and \eqref{dual_SVM2}, respectively.
		
		To complete the proof of the lemma, it suffices to show the following inequality:
		\begin{equation} \label{dual_ineq}
			\sum_{i,j=1}^n \alpha_i^* \beta_j^* k(a_i,b_j) \geq {1 \over 4} \exp(-100\log n \cdot (t-1)).
		\end{equation}
		
		Notice that so far we did not use the fact that there is a close pair of vectors $a_i \in A$ and $b_j \in B$ such that $\Hamming(a_i,b_j)\leq t-1$. We use this fact now. We lower bound the left hand side of \eqref{dual_ineq} by the summand corresponding to the close pair:
		$$
			\sum_{i,j=1}^n \alpha_i^* \beta_j^* k(a_i,b_j) \geq \alpha_i^* \beta_j^* k(a_i,b_j) \geq \alpha_i^* \beta_j^* \exp(-100\log n \cdot (t-1)),
		$$
		where in the last inequality we use $\Hamming(a_i,b_j)\leq t-1$ and the definition of the Gaussian kernel.
		
		The proof is completed by observing that $\alpha_i^* \geq {1 \over 2}$ and $\beta_i^* \geq {1 \over 2}$ which follows from Fact \ref{primal_dual} and Lemma \ref{bounded}.
	\end{proof}

\section{Hardness for SVM with the bias term} \label{section_with_bias}
	In the previous section we showed hardness for SVM without the bias term. In this section we show hardness for SVM with the bias term.

	\begin{theorem} \label{hardness_bias}
		Let $x_1, \ldots, x_n \in \{-1,0,1\}^d$ be $n$ vectors and let $y_1, \ldots, y_n \in \{-1,1\}$ be $n$ labels.
		
		Let $k(a,a')=\exp\left(-C\|a-a'\|_2^2\right)$ be the Gaussian kernel with $C=100\log n$.
		
		Consider the corresponding hard-margin SVM with the bias term:
		\begin{equation} \label{eq_primal_hardness_bias}
			\begin{aligned}
				& \underset{\alpha_1, \ldots, \alpha_n \geq 0, \ \ b}{\text{minimize}}
				& & {1 \over 2}\sum_{i,j=1}^n \alpha_i \alpha_j y_i y_j k(x_i,x_j) \\
				& \text{subject to}
				& & y_i f(x_i) \geq 1, \ \ i = 1, \ldots, n,
			\end{aligned}
		\end{equation}
		where $f(x):=b+\sum_{i=1}^n \alpha_i y_i k(x_i,x)$.
		
		Consider any $\eps=\exp(-\omega(\log^2 n))$. 
		Approximating the optimal value of \eqref{eq_primal_hardness_bias} within the multiplicative factor $(1+\eps)$ requires almost quadratic time assuming SETH.
		This holds for the dimensionality $d=O(\log^3 n)$ of the input vectors.
		
		The same hardness result holds for any additive $\exp(-\omega(\log^2 n))$ approximation factor.
	\end{theorem}
	\begin{proof}
		Consider a hard instance from Theorem \ref{hardness} for SVM without the bias term.
		Let $x_1, \ldots, x_n \in \{0,1\}^{d}$ be the $n$ binary vectors of dimensionality $d=\omega(\log n)$ and $y_1, \ldots, y_n \in \{-1,1\}$ be the $n$ corresponding labels.
		For this input consider the dual SVM without the bias term (see Definition~\ref{def_dual}):
		\begin{equation} \label{dSVM}
			\begin{aligned}
				& \underset{\gamma_1, \ldots, \gamma_n \geq 0}{\text{maximize}}
				& & \sum_{i=1}^n \gamma_i \ - \ {1 \over 2}\sum_{i,j=1}^n \gamma_i \gamma_j y_i y_j k(x_i,x_j).
			\end{aligned}
		\end{equation}
	
		We will show how to reduce SVM without the bias term \eqref{dSVM} to SVM with the bias term. By Theorem \ref{hardness} this will give hardness result for SVM with the bias term. We start with a simpler reduction that will achieve almost what we need except the entries of the vectors will not be from the set $\{-1,0,1\}$. Then we will show how to change the reduction to fix this.
		
		Consider $2n$ vectors $x_1, \ldots, x_n, -x_1, \ldots, -x_n \in \{-1,0,1\}^d$ with $2n$ labels $y_1, \ldots, y_n, -y_1, \ldots, -y_n \in \{-1, 1\}$. Consider an SVM with the bias term for the $2n$ vectors, that is, an SVM of the form \eqref{eq_primal_hardness_bias}. From Definition \ref{def_dual} we know that its dual SVM is
		\begin{equation} \label{eq_dual_bias}
			\begin{aligned}
				& \underset{\substack{\alpha_1, \ldots, \alpha_n \geq 0\\ \beta_1, \ldots, \beta_n \geq 0}}{\text{maximize}}
				& & \sum_{i=1}^n \alpha_i \ + \ \sum_{j=1}^n \beta_j \ \\
				& & & - \ {1 \over 2}\sum_{i,j=1}^n \alpha_i \alpha_j y_i y_j k(x_i,x_j) \ - \ {1 \over 2}\sum_{i,j=1}^n \beta_i \beta_j y_i y_j k(x_i,x_j) \ + \ \sum_{i,j=1}^n \alpha_i \beta_j y_i y_j k(x_i,-x_j) \\
				& \text{subject to}
				& & \sum_{i=1}^n \alpha_i y_i=\sum_{j=1}^n \beta_j y_j.
			\end{aligned}
		\end{equation}
		
		Consider any setting of values for $\alpha_i$ and $\beta_j$. Notice that if we swap the value of $\alpha_i$ and $\beta_i$ for every $i$, the value of the objective function of \eqref{eq_dual_bias} does not change. This is implies that we can define $\gamma_i:={\alpha_i+\beta_i \over 2}$ and set $\alpha_i=\beta_i=\gamma_i$ for every $i$. Because of the convexity of the optimization problem, the value of the objective function can only increase after this change. Clearly, the equality constraint will be satisfied. Therefore, w.l.o.g. we can assume that $\alpha_i=\beta_i=\gamma_i$ for some $\gamma_i$ and we can omit the equality constraint.
		
		We rewrite \eqref{eq_dual_bias} in terms of $\gamma_i$ and divide the objective function by $2$:
		\begin{equation} \label{eq_dual_bias_1}
			\begin{aligned}
				& \underset{\gamma_1, \ldots, \gamma_n \geq 0}{\text{maximize}}
				& & \sum_{i=1}^n \gamma_i \ - \ {1 \over 2}\sum_{i,j=1}^n \gamma_i \gamma_j y_i y_j k(x_i,x_j) \ + \ {1 \over 2}\sum_{i,j=1}^n \gamma_i \gamma_j y_i y_j k(x_i,-x_j).
			\end{aligned}
		\end{equation}
		
		Notice that \eqref{eq_dual_bias_1} and \eqref{dSVM} are almost the same. The only difference is the third term $${1 \over 2}\sum_{i,j=1}^n \gamma_i \gamma_j y_i y_j k(x_i,-x_j)$$ in \eqref{eq_dual_bias_1}. We can make this term to be equal to $0$ and not change the first two terms as follows. We append an extra coordinate to every vector $x_i$ and set this coordinate to be large enough value $M$. If we set $M=+\infty$, the third term becomes $0$. The first term does not depend on the vectors. The second term depends only on the distances between the vectors (which are not affected by adding the same entry to all vectors). Thus, the first two terms do not change after this modification.
		
		We showed that we can reduce SVM without the bias term \eqref{dSVM} to the SVM with the bias term \eqref{eq_dual_bias}. By combining this reduction with Theorem \ref{hardness} we obtain hardness for SVM with the bias term. This is almost what we need except that the reduction presented above produces vectors with entries that are not from the set $\{-1,0,1\}$. In every vector $x_i$ or $-x_i$ there is an entry that has value $M$ or $-M$, respectively. In the rest of the proof we show how to fix this, by bounding $M$ by $O(\log^3 n)$ and distributing its contribution over  $O(\log^3 n)$ coordinates. 
		
		\paragraph{Final reduction} The final reduction is as follows:
		\begin{itemize}
			\item Take a hard instance for the SVM without the bias term from Theorem \ref{hardness}. Let $x_1, \ldots, x_n \in \{0,1\}^{d}$ be the $n$ binary vectors of dimensionality $d=\omega(\log n)$ and $y_1, \ldots, y_n \in \{-1,1\}$ be the $n$ corresponding labels.
			\item Append $\log^3 n$ entries to each of the vectors $x_i$, $i=1, \ldots, n$ and set the entries to be $1$.
			\item Solve SVM \eqref{eq_primal_hardness_bias} on the $2n$ vectors $x_1, \ldots, x_n, -x_1, \ldots, -x_n \in \{-1,0,1\}^d$ with $2n$ labels $y_1, \ldots, y_n, -y_1, \ldots, -y_n \in \{-1, 1\}$. Let $V$ be the optimal value of the SVM divided by $2$.
			\item Output $V$.
		\end{itemize}
		
		\paragraph{Correctness of the reduction} From the above discussion we know that we output the optimal value $V$ of the optimization problem \eqref{eq_dual_bias_1}. Let $V'$ be the optimal value of SVM \eqref{dSVM}.
		
		By Theorem \ref{hardness}, it is sufficient to show that $|V-V'|\leq \exp(-\omega(\log^2 n))$ to establish hardness for SVM with the bias term. We will show that $|V-V'|\leq n^{O(1)}\exp(-\log^3 n)$. This gives hardness for additive approximation of SVM with the bias term. However, $|V-V'|\leq \exp(-\omega(\log^2 n))$ is also sufficient to show hardness for multiplicative approximation (see the discussion on the approximation in the proof of Theorem \ref{hardness}).
		
		In the rest of the section we show that $|V-V'|\leq n^{O(1)}\exp(-\log^3 n)$. Let $\gamma_i'$ be the assignment to $\gamma_i$ that achieves $V'$ in SVM \eqref{dSVM}. Let $\gamma_i^*$ be the assignment to $\gamma_i$ that achieves $V$ in \eqref{eq_dual_bias_1}. We will show that $\gamma_i'\leq O(n)$ for all $i=1, \ldots, n$. It is also true that $\gamma_i^*\leq O(n)$ for all $i=1, \ldots, n$ and the proof is analogous. Since $x_1, \ldots, x_n$ are different binary vectors and $k(x_i,x_j)$ is the Gaussian kernel with the parameter $C=100\log n$, we have that $k(x_i,x_j)\leq 1/n^{10}$ for all $i \neq j$. This gives the following upper bound:
		$$
			V'=\sum_{i=1}^n \gamma_i' \ - \ {1 \over 2}\sum_{i,j=1}^n \gamma_i' \gamma_j' y_i y_j k(x_i,x_j) \leq \sum_{i=1}^n\left(\gamma_i' - \left({1 \over 2}-o(1)\right)(\gamma_i')^2\right).
		$$
		Observe that every non-negative summand on the right hand side is at most $O(1)$. Therefore, if there exists $i$ such that $\gamma_i'\geq \omega(n)$, then the right hand side is negative. This contradicts the lower bound $V'\geq 0$ (which follows by setting all $\gamma_i$ to be $0$ in \eqref{dSVM}).
		
		By plugging $\gamma_i'$ into \eqref{eq_dual_bias_1} and using the fact that $\gamma_i'\leq O(n)$, we obtain the following inequality:
		\begin{equation} \label{VV1}
			V \geq V'+{1 \over 2}\sum_{i,j=1}^n \gamma_i' \gamma_j' y_i y_j k(x_i,-x_j) \geq V'-n^{O(1)}\exp(-\log^3 n).
		\end{equation}
		In the last inequality we use $k(x_i,-x_j)\leq \exp(-\log^3 n)$ which holds for all $i,j=1, ..., n$ (observe that each $x_i$ and $x_j$ ends with $\log^3 n$ entries $1$ and use the definition of the Gaussian kernel).
		
		Similarly, by plugging $\gamma_i^*$ into \eqref{dSVM} and using the fact that $\gamma_i^*\leq O(n)$, we obtain the following inequality:
		\begin{equation} \label{VV2}
			V' \geq V-{1 \over 2}\sum_{i,j=1}^n \gamma_i^* \gamma_j^* y_i y_j k(x_i,-x_j) \geq V-n^{O(1)}\exp(-\log^3 n).
		\end{equation}
		
		Inequalities \eqref{VV1} and \eqref{VV2} combined give the desired inequality $|V-V'|\leq n^{O(1)}\exp(-\log^3 n)$.
	\end{proof}
	
\section{Hardness for soft-margin SVM}
\label{s:soft}
	\begin{theorem} \label{hardness_soft}
		Let $x_1, \ldots, x_n \in \{-1,0,1\}^d$ be $n$ vectors and let $y_1, \ldots, y_n \in \{-1,1\}$ be $n$ labels.
		
		Let $k(a,a')=\exp\left(-C\|a-a'\|_2^2\right)$ be the Gaussian kernel with $C=100\log n$.
		
		Consider the corresponding soft-margin SVM with the bias term:
		\begin{equation} \label{eq_hardness_soft}
			\begin{aligned}
				& \underset{\substack{\alpha_1, \ldots, \alpha_n \geq 0, \ \ b \\ \xi_1, \ldots, \xi_n \geq 0}}{\text{minimize}}
				& & {\lambda \over 2}\sum_{i,j=1}^n \alpha_i \alpha_j y_i y_j k(x_i,x_j) \ + \ {1 \over n}\sum_{i=1}^n \xi_i\\
				& \text{subject to}
				& & y_i f(x_i) \geq 1-\xi_i, \ \ i = 1, \ldots, n,
			\end{aligned}
		\end{equation}
		where $f(x):=b+\sum_{i=1}^n \alpha_i y_i k(x_i,x)$.
		
		Consider any $\eps=\exp(-\omega(\log^2 n))$ and any $0<\lambda\leq {1 \over Kn^2}$ for a large enough constant $K>0$.
		Approximating the optimal value of \eqref{eq_hardness_soft} within the multiplicative factor $(1+\eps)$ requires almost quadratic time assuming SETH.
		This holds for the dimensionality $d=O(\log^3 n)$ of the input vectors.
		
		The same hardness result holds for any additive $\exp(-\omega(\log^2 n))$ approximation factor.
	\end{theorem}
	\begin{proof}
		Consider the hard instance from Theorem \ref{hardness_bias} for the hard-margin SVM. The dual of the hard-margin SVM is \eqref{eq_dual_bias}. From the proof we know that the optimal $\alpha_i$ and $\beta_i$ satisfy $\alpha_i=\beta_i=\gamma_i^*\leq 2Kn$ for some large enough constant $K>0$ for all $i=1, \ldots, n$. Thus, w.l.o.g.\ we can add these inequalities to the set of constraints.
		We compare the resulting dual SVM to Definition~\ref{def_soft} and conclude that the resulting dual SVM is a dual of a \emph{soft-margin} SVM with the regularization parameter $\lambda={1 \over Kn^2}$. Therefore, the hardness follows from Theorem \ref{hardness_bias}.
	\end{proof}
 
\section{Hardness for kernel ridge regression}

We start with stating helpful definitions and lemmas.

We will use the following lemma which is a consequence of the binomial inverse theorem.

\begin{lemma} \label{inverse}
	Let $X$ and $Y$ be two square matrices of equal size. Then the following equality holds:
	$$
		(X+Y)^{-1}=X^{-1}-X^{-1}(I+YX^{-1})^{-1}YX^{-1}.
	$$
\end{lemma}

\begin{definition}[Almost identity matrix] \label{almost_identity}
	Let $X \in \R^{n \times n}$ be a matrix. We call it \emph{almost identity matrix} if $X=I+Y$ and $|Y_{i,j}|\leq n^{-\omega(1)}$ for all $i,j=1, \ldots, n$.
\end{definition}

We will need the following two lemmas.

\begin{lemma} \label{identity_product}
	The product of two almost identity matrices is an almost identity matrix.
\end{lemma}
\begin{proof} Follows easily from the definition.
\end{proof}

\begin{lemma} \label{identity_inverse}
	The inverse of an almost identity matrix is an almost identity matrix.
\end{lemma}
\begin{proof}
	Let $X$ be an almost identity matrix. We want to show that $X^{-1}$ is an almost identity matrix. We write $X=I-Y$ such that $|Y_{i,j}|\leq n^{-\omega(1)}$ for all $i,j=1, \ldots, n$. We have the following matrix equality
	$$
		X^{-1}=(I-Y)^{-1}=I+Y+Y^2+Y^3+\ldots
	$$
	To show that $X^{-1}$ is an almost identity, we will show that the absolute value of every entry of $Z:=Y+Y^2+Y^3+\ldots$ is at most $n^{-\omega(1)}$. Let $\eps\leq n^{-\omega(1)}$ is an upper bound on $|Y_{i,j}|$ for all $i,j=1, \ldots, n$. Then $|Z_{i,j}|\leq Z_{i,j}'$, where $Z':=Y'+(Y')^2+(Y')^3+\ldots$ and $Y'$ is a matrix consisting of entries that are all equal to $\eps$. The proof follows since $Z_{i,j}'=\sum_{k=1}^{\infty}\eps^k n^{k-1}\leq 10\eps\leq n^{-\omega(1)}$.
\end{proof}

In the rest of the section we prove Theorem \ref{ridge_regression}.
\begin{proof}[Proof of Theorem \ref{ridge_regression}]
	We reduce the BHCP problem to the problem of computing the sum of the entries of $K^{-1}$.
	
	Let $A$ and $B$ be the two sets of binary vectors from the BHCP instance. Let $K \in \R^{2n \times 2n}$ be the corresponding kernel matrix. We can write the kernel matrix $K$ as combination of four smaller matrices $K^{1,1}, K^{1,2}, K^{2,1}, K^{2,2} \in \R^{n \times n}$:
	$$
		K=
			\left[
			\begin{array}{c|c}
			K^{1,1} & K^{1,2} \\
			\hline
			K^{2,1} & K^{2,2}
			\end{array}
			\right].
	$$
	$K^{1,1}$ is the kernel matrix for the set of vectors $A$ and $K^{2,2}$ is the kernel matrix for the set of vectors $B$. We define two new matrices $X, Y \in \R^{2n \times 2n}$: 
		$X=\left[
			\begin{array}{c|c}
			K^{1,1} & 0 \\
			\hline
			0 & K^{2,2}
			\end{array}
			\right]
		$ and
		$Y=\left[
			\begin{array}{c|c}
			0 & K^{1,2} \\
			\hline
			K^{2,1} & 0
			\end{array}
			\right]
		$.
	
	For any matrix $Z$, let $s(Z)$ denote the sum of all entries of $Z$. Using Lemma \ref{inverse}, we can write $K^{-1}$ as follows:
	$$
		K^{-1}=(X+Y)^{-1}=X^{-1}-X^{-1}(I+YX^{-1})^{-1}YX^{-1}.
	$$
	We note that the matrix $X$ is an almost identity and that $|Y_{i,j}|\leq n^{-\omega(1)}$ for all $i,j=1, \ldots, 2n$. This follows from the fact that we use the Gaussian kernel function with the parameter $C=\omega(\log n)$ and the input vectors are binary. Combining this with lemmas \ref{identity_product} and \ref{identity_inverse} allows us to conclude that matrices $X^{-1}(I+YX^{-1})^{-1}$ and $X^{-1}$ are almost identity. Since all entries of the matrix $Y$ are non-negative, we conclude that 
	$$
		s(X^{-1}(I+YX^{-1})^{-1}YX^{-1})=s(Y)(1 \pm n^{-\omega(1)}).
	$$
	We obtain that
	\begin{align*}
		s(K^{-1})& =s(X^{-1})-s(X^{-1}(I+YX^{-1})^{-1}YX^{-1})\\
		& =s(X^{-1})-s(Y)(1 \pm n^{-\omega(1)})\\
		& =s\left((K^{1,1})^{-1}\right)+s\left((K^{2,2})^{-1}\right)-s(Y)(1 \pm n^{-\omega(1)}).
	\end{align*}
	
	Fix any $\alpha=\exp(-\omega(\log^2 n))$. Suppose that we can estimate each $s(K^{-1})$, $s\left((K^{1,1})^{-1}\right)$ and $s\left((K^{2,2})^{-1}\right)$ within the additive factor of $\alpha$.
	This allows us to estimate $s(Y)$ within the additive factor of $10\alpha$.
	This is enough to solve the BHCP problem. We consider two cases.
	
	\paragraph{Case 1} There are no close pairs, that is, for all $i,j=1, \ldots, n$ we have $||a_i-b_j||_2^2\geq t$ and $\exp(-C||a_i-b_j||_2^2)\leq \exp(-Ct)=:\delta$. Then $s(Y)\leq 2n^2\delta$.
	
	\paragraph{Case 2} There is a close pair. That is, $||a_{i'}-b_{j'}||_2^2\leq t-1$ for some $i',j'$. This implies that $\exp(-C||a_{i'}-b_{j'}||_2^2)\geq \exp(-C(t-1))=:\Delta$. Thus, $s(Y)\geq \Delta$.
	
	Since $C=\omega(\log n)$, we have that $\Delta\geq 100n^{2}\delta$ and we can distinguish the two cases assuming that the additive precision $\alpha=\exp(-\omega(\log^2 n))$ is small enough.
	
	\paragraph{Precision} To distinguish $s(Y)\leq 2n^2 \delta$ from $s(Y)\geq \Delta$, it is sufficient that $\Delta\geq 100n^2\delta$ and $\alpha\leq \Delta/1000$. We know that $\Delta\geq 100n^2\delta$ holds because $C=\omega(\log n)$. Since $\Delta\leq \exp(-Cd)$, we want to choose $C$ and $d$ such that the $\alpha\leq \Delta/1000$ is satisfied. We can do that because we can pick $C$ to be any $C=\omega(\log n)$ and the BHCP problem requires almost quadratic time assuming SETH for any $d=\omega(\log n)$.
	
	We get that additive $\eps$ approximation is sufficient to distinguish the cases for any $\eps=\exp(-\omega(\log^2 n))$.
	We observe that $s(K^{-1})\leq O(n)$ for any almost identity matrix $K$. This means that $(1+\eps)$ multiplicative approximation is sufficient for the same $\eps$. This completes the proof of the theorem.
\end{proof}
 \section{Hardness for training of the final layer of a neural network}

Recall that the trainable parameters are $\alpha:=(\alpha_1, \ldots, \alpha_n)^\transp$, and that the optimization problem \eqref{opt1} is equivalent to the following optimization problem:
\begin{equation} \label{opt2}
	\begin{aligned}
		& \underset{\alpha \in \R^n}{\text{minimize}}
		& & \sum_{i=1}^m l(y_i \cdot (M\alpha)_i),
	\end{aligned}
\end{equation}
where $M \in \R^{m \times n}$ is the matrix defined as $M_{i,j}:=S(a_i^\transp b_j)$ for $i=1, \ldots, m$ and $j=1, \ldots n$.
For the rest of the section we will use $m=\Theta(n)$.

Let $A=\{a_1, \ldots, a_n\} \subseteq \{0,1\}^d$ and $B=\{b_1, \ldots, b_n\} \subseteq \{0,1\}^d$ with $d=\omega(\log n)$ be the input to the Orthogonal Vectors problem. We construct a matrix $M$ as a combination of three smaller matrices:
$$
	M=
		\left[
			\begin{array}{c}
				M_1 \\
				M_2 \\
				M_2
			\end{array}
		\right].
$$
Both matrices $M_1, M_2 \in \R^{n \times n}$ are of size $n \times n$. Thus we have that the number of rows of $M$ is $m=3n$. 

We describe the two matrices $M_1, M_2$ below.
Recall that $v_0$, $v_1$, and $v_2$ are given in Definition \ref{activation}.
\begin{itemize}
	\item $(M_1)_{i,j}=S\left(v_0-(v_2-v_0)\cdot a_i^\transp b_j\right)$.
    For any two real values $x, y \in \R$ we write $x \approx y$ if $x=y$ up to an inversely superpolynomial additive factor. In other words, $|x-y|\leq n^{-\omega(1)}$. We observe that if two vectors $a_i$ and $b_j$ are orthogonal, then the corresponding entry $(M_1)_{i,j}=S(v_0)=\Theta(1)$ and otherwise $(M_1)_{i,j}\approx 0$. We will show that an $\left( 1+{1 \over 4n} \right)$-approximation of the optimal value of the optimization problem \eqref{opt2} will allow us to decide whether there is an entry in $M_1$ that is $S(v_0)=\Theta(1)$. This will give the required hardness.

    It remains to show how to construct the matrix $M_1$ using a neural network.
    We set the weights for the $j$-th hidden unit to be $\left[\begin{array}{c}b_j \\1\end{array}\right]$.
    That is, $d$ weights are specified by the vector $b_j$, and we add one more input with weight $1$.
    The $i$-th example (corresponding to the $i$-th row of the matrix $M_1$) is the vector $\left[\begin{array}{c}-(v_2-v_0)a_i \\v_0\end{array}\right]$.
    The output of the $j$-th unit on this example (which corresponds to entry $(M_1)_{i,j}$) is equal to 
	\begin{align*}
    S\left(\left[\begin{array}{c}-(v_2-v_0)a_i \\v_0\end{array}\right]^\transp \left[\begin{array}{c}b_j \\1\end{array}\right]\right) \; &= \; S\left(v_0-(v_2-v_0)\cdot a_i^\transp b_j\right)\\
		&= \; (M_1)_{i,j}
	\end{align*}
	as required.
	\item $(M_2)_{i,j}=S\left(v_1-(v_2-v_1)\cdot \bar{b_i}^\transp b_j\right)$, where $\bar{b_i}$ is a binary vector obtained from the binary vector $b_i$ by complementing all bits. We observe that this forces the diagonal entries of $M_2$ to be equal to $(M_2)_{i,i}=S(v_1)=1/n^{1000K}$ for all $i=1, \ldots, n$ and the off-diagonal entries to be $(M_2)_{i,j}\approx 0$ for all $i \neq j$.\footnote{For all $i \neq j$ we have $\bar{b_i}^\transp  b_j\geq 1$. This holds because all vectors $b_i$ are distinct and have the same number of 1s. 
	}
\end{itemize}
	
	To complete the description of the optimization problem \eqref{opt2}, we assign labels to the inputs corresponding to the rows of the matrix $M$.
We assign label $1$ to all inputs corresponding to rows of the matrix $M_1$ and the first copy of the matrix $M_2$.
We assign label $-1$ to all remaining rows of the matrix $M$ corresponding to the second copy of matrix $M_2$.
	
It now suffices to prove Lemma \ref{upper_bound} and Lemma \ref{lower_bound}.

	\begin{proof}[Proof of Lemma \ref{upper_bound}]
		To obtain an upper bound on the optimal value in the presence of an orthogonal pair, we set the vector $\alpha$ to have all entries equal to $n^{100K}$. For this $\alpha$ we have 
		\begin{itemize}
			\item $\left|(M_1 \alpha)_i\right|\geq \Omega(n^{100K})$ for all $i=1, \ldots, n$ such that there is exists $j=1, \ldots, n$ with $a_i^\transp b_j=0$. Let $x \geq 1$ be the number of such $i$.
			\item $\left|(M_1 \alpha)_i\right|\leq n^{-\omega(1)}$ for all $i=1, \ldots, n$ such that there is no $j=1, \ldots, n$ with $a_i^\transp b_j=0$. The number of such $i$ is $n-x$.
		\end{itemize}
		By using the second property of Definition \ref{nice}, the total loss corresponding to $M_1$ is upper bounded by 
		\begin{align*}
			x \cdot l(\Omega(n^{100K})) + (n-x) \cdot l(n^{-\omega(1)}) & \leq x \cdot o(1) + (n-x) \cdot (l(0)+o(1/n))\\
			& \leq (n-1) \cdot l(0)+o(1)=:l_1.
		\end{align*}
		Finally, the total loss corresponding to the two copies of the matrix $M_2$ is upper bounded by 
		\begin{align*}
			2n \cdot l(\pm O(n^{-800K})) & = 2n \cdot (l(0) \pm o(1/n)) \\
			& \leq 2n \cdot l(0)+ o(1) =: l_2.
		\end{align*}
		The total loss corresponding to the matrix $M$ is upper bounded by $l_1+l_2\leq (3n-1)\cdot l(0)+o(1)$ as required.
	\end{proof}
	
	\begin{proof}[Proof of Lemma \ref{lower_bound}]
		We first observe that the total loss corresponding to the two copies of the matrix $M_2$ is lower bounded by $2n \cdot l(0)$. Consider the $i$-th row in both copies of matrix $M_2$. By using the convexity of the function $l$, the loss corresponding to the two rows is lower bounded by $l((M_2 \alpha)_i)+l(-(M_2 \alpha)_i) \geq 2\cdot l(0)$. By summing over all $n$ pairs of rows we obtain the required lower bound on the loss.
		
		We claim that $\|\alpha\|_{\infty}\leq n^{10^6 K}$. Suppose that this is not the case and let $i$ be the index of the largest entry of $\alpha$ in magnitude. Then the $i$-th entry of the vector $M_2 \alpha$ is 
		\begin{align*}
			(M_2 \alpha)_i \; & = \; \alpha_i (M_2)_{i,i} \, \pm \, n \cdot \alpha_i \cdot n^{-\omega(1)}\\
			& \geq \; {\alpha_i \over n^{1000K}} - \alpha_i n^{-\omega(1)} \; ,
		\end{align*}
		where we recall that the diagonal entries of matrix $M_2$ are equal to $(M_2)_{i,i}=S(v_1)=1/n^K$. If $|\alpha_i|>n^{10^6K}$, then $|(M_2\alpha)_i|\geq n^{1000k}$. However, by the second property in Definition \ref{nice}, this implies that the loss is lower bounded by $\omega(n)$ for the $i$-row (for the first or the second copy of $M_2$). This contradicts a simple lower bound of $4n\cdot l(0)$ on the loss obtained by setting $\alpha=0$ to be the all $0$s vector. We use the third property of a nice loss function which says that $l(0)>0$.
		
		For the rest of the proof, we assume that $\|\alpha\|_{\infty}\leq n^{10^6 K}$. We will show that the total loss corresponding to $M_1$ is lower bounded by $n\cdot l(0)-o(1)$. This is sufficient since we already showed that the two copies of $M_2$ contribute a loss of at least $2n \cdot l(0)$.

		Since all entries of the matrix $M_1$ are inversely superpolynomial (there is no pair of orthogonal vectors), we have that $|(M_1\alpha)_i|\leq n^{-\omega(1)}$ for all $i=1, \ldots, n$. Using the second property again, the loss corresponding to $M_1$ is lower bounded by 
		\begin{align*}
			n \cdot l(\pm n^{-\omega(1)})& \geq n \cdot (l(0)-o(1/n))\\
			& \geq n\cdot l(0)-o(1)
		\end{align*}
		as required.
	\end{proof}
 	
\end{document}